\documentclass[letterpaper, 10 pt, conference]{ieeeconf}  % Comment this line out if you need a4paper

\IEEEoverridecommandlockouts                              % This command is only needed if 
\UseRawInputEncoding

\usepackage{graphics} % for pdf, bitmapped graphics files
\usepackage{epsfig} % for postscript graphics files
\usepackage{mathptmx} % assumes new font selection scheme installed
\usepackage{times} % assumes new font selection scheme installed
\usepackage{amsmath} % assumes amsmath package installed
\usepackage{amssymb}  % assumes amsmath package installed
\usepackage{amsmath,nicefrac}
\usepackage{float}
\usepackage{multicol}
\usepackage{epstopdf}
\usepackage{graphicx}
\usepackage{dblfloatfix}
\usepackage{subcaption}
\usepackage{array}
\usepackage{cite}
\newcolumntype{P}[1]{>{\centering\arraybackslash}p{#1}}
\usepackage{balance}
\usepackage{xcolor}
\usepackage{comment}

\usepackage{rotating}
\usepackage{wrapfig}
\usepackage[utf8]{inputenc}
\usepackage[english]{babel}
\usepackage{amsthm}
\newtheorem{theorem}{Theorem}

\newtheorem{lemma}[theorem]{Lemma}
\theoremstyle{definition}
\newtheorem{definition}{Definition}[]

\newtheorem{remark}{Remark}

\usepackage{cite}
\usepackage{amsmath,amssymb,amsfonts}
\usepackage{algorithmic}
\usepackage{graphicx}

\usepackage{psfrag}
\usepackage{cancel}

\makeatletter
\let\NAT@parse\undefined
\makeatother
\usepackage{hyperref}

\makeatletter
\newcommand{\linebreakand}{%
  \end{@IEEEauthorhalign}
  \hfill\mbox{}\par
  \mbox{}\hfill\begin{@IEEEauthorhalign}
}
\makeatother

\def\mc{\mathcal}

\title{\LARGE \bf
A negative imaginary approach to hybrid integrator-gain system control
}

\author{{Kanghong Shi$^{\dagger}$}\\
\textit{kanghong.shi@anu.edu.au}
\and
{Nastaran Nikooienejad$^{\dagger}$}\\
\textit{nastaran.nikooienejad@utdallas.edu}\vspace{0.5cm}
\linebreakand
{Ian R. Petersen, \IEEEmembership{Fellow, IEEE}}\\
\textit{ian.petersen@anu.edu.au}
\and
{S. O. Reza Moheimani, \IEEEmembership{Fellow, IEEE}}\\
\textit{reza.moheimani@utdallas.edu}
\thanks{$^{\dagger}$ Equal contribution.}
\thanks{This work was supported by the Australian Research Council under grant DP190102158, and partially by the UTD Center for Atomically Precise Fabrication of Solid-state Quantum Devices.}
\thanks{Kanghong Shi and Ian R. Petersen are with the School of Engineering, College of Engineering and Computer Science, Australian National University, Canberra, Acton, ACT 2601, Australia. Nastaran Nikooienejad and S. O. Reza Moheimani  are with the Erik Jonsson School of Engineering and Computer Science, The University of Texas at Dallas, Richardson, TX 75080 USA. Corresponding author: S. O. Reza Moheimani.}%
}

\begin{document}

\maketitle
\thispagestyle{empty}
\pagestyle{empty}

	%%%%%%%%%%%%%%%%%%%%%%%%%%%%%%%%%%%%%%%%%%%%%%%%%%%%%%%%%%%%%%%%%%%%%%%%%%%%%%%%
	\begin{abstract}
		In this paper, we show that a hybrid integrator-gain system (HIGS) is a nonlinear negative imaginary (NNI) system. We prove that the positive feedback interconnection of a linear negative imaginary (NI) system and a HIGS is asymptotically stable. We apply the HIGS to a MEMS nanopositioner, as an example of a linear NI system, in a single-input single-output framework. We analyze the stability and the performance of the closed-loop interconnection in both time and frequency domains through simulations and demonstrate the applicability of HIGS as an NNI controller to a linear NI system.
	\end{abstract}
	

	%%%%%%%%%%%%%%%%%%%%%%%%%%%%%%%%%%%%%%%%%%%%%%%%%%%%%%%%%%%%%%%%%%%%%%%%%%%%%%%%
	\section{INTRODUCTION}
	The hybrid integrator-gain system known as HIGS was introduced in~\cite{Deenen_HIGS_motion_control_2017} as a nonlinear integrator system to overcome the limitations of Bode's gain-phase relationship~\cite{Heertjes_HIGS_Vibration_2019}. The HIGS inherits the potential benefits of typical reset control systems such as the Clegg Integrator (CI)~\cite{Clegg_1958} and the First Order Reset Element (FORE) system~\cite{Zaccarian_FORE_2005} enhancing the phase lag by about $52$ degrees compared to a standard integrator, without the common drawbacks of reset systems~\cite{Deenen_HIGS_motion_control_2017}. This hybrid system alternates between integrator and gain modes instead of resetting the state to zero, which yields a non-smooth but continuous control signal and prevents excitation of higher harmonics induced by reset. The control signal, governed by switching logic, also satisfies a sector constraint that restricts the input-output behavior of the HIGS element and enforces the HIGS input and output to be of equal sign~\cite{Heertjes_HIGS_Vibration_2019}.
	
	Control design based on HIGS has found applications in high-precision mechatronic systems for motion tracking~\cite{Deenen_HIGS_motion_control_2017, Eijnden_HIGS_motion_control_2018, Gruntjens_HIGS_Lens_motion_2019} and vibration isolation and damping~\cite{Heertjes_HIGS_Vibration_2019,Achten_HIGS_Skyhook_thesis_2020,Baaij_HIGS_positive_real_systems}. The HIGS element is exploited in various control configurations with linear controllers. For instance, a HIGS-based $\text{PI}^2\text{D}$ controller is designed and employed in~\cite{Deenen_HIGS_motion_control_2017} to control a wafer stage system of an industrial wafer scanner.  With improved phase behavior compared to typical low-pass filters, a HIGS-based second-order low-pass filter is proposed in~\cite{Eijnden_HIGS_motion_control_2018} to achieve a substantial low-frequency disturbance rejection in an industrial wafer scanner. A HIGS element is cascaded with another HIGS element in~\cite{Heertjes_HIGS_Vibration_2019} to construct a HIGS-based band-pass filter for vibration isolation. As discussed in~\cite{Eijnden_HIGS_Overshoot_limitation_2020}, the hybrid integrator-gain system can deal with inherent design limitations in systems that contain an unstable pole by replacing a standard integrator with the HIGS element.
	
	Stability and performance analysis of the closed-loop system with the HIGS is challenging due to the nonlinear nature of the HIGS element. For stability analysis, the closed-loop system is rearranged in Lur'e form where the nonlinearity is isolated from the linear part of the system~\cite{Deenen_HIGS_motion_control_2017}. Accordingly, the input-to-state stability (ISS) of the closed-loop system is guaranteed based on the detectability of the HIGS element and the  circle criterion whereby the sector-boundedness of the nonlinear element along with a loop transformation allows for the application of the passivity theorem~\cite{Vanloon_stability_reset_system_2017}. Accordingly, it is assumed that the underlying linear system satisfies the circle-criterion condition which is less stringent than the strictly positive real criterion~\cite{Loon_stability_reset_2017}.  In~\cite{Deenen_HIGS_motion_control_2017}, stability analysis of the closed-loop system is transformed to frequency-domain conditions that are graphically verifiable using the frequency response function (FRF) of the linear part of the system in the Lur'e form. The stability conditions in~\cite{Deenen_HIGS_motion_control_2017} imply that the underlying linear system must be Hurwitz which is not the case with some types of systems. Therefore, a novel frequency-domain stability analysis is proposed in~\cite{Eijnden_frequency_stability_HIGS_2021} which provides a less conservative stability criterion and incorporates the dynamic nature of the HIGS system. Moreover, less conservative stability conditions are presented in~\cite{Deenen_projection_based_2022} based on a piecewise quadratic Lyapunov function.
	
	Application of the HIGS element to flexible structures with collocated force actuators and position sensors would also be interesting where the underlying system is negative imaginary (NI).
	%%%%%%%%%%%%%%%%%%%%%%%%%%%%%%%%%%%%%%%%%%%%%%%%%%%%%%%%%%%%%%%%%%%%%%%%%%%
    NI systems theory was introduced in~\cite{Lanzon_stability_2008,Petersen_feedback_2010} to address challenges confronting vibration control of flexible structures \cite{preumont2018vibration,halim2001spatial,pota2002resonant}. Such systems often have highly resonant dynamics, which makes negative velocity feedback control unreliable. NI systems theory provides an alternative control approach, which uses positive position feedback control. From this point of view, NI systems theory can be regarded as a complement to the positive real (PR) systems theory. While a PR system can only have relative degree zero or one, an NI system can have relative degree zero, one, and two \cite{brogliato2007dissipative,shi2021necessary}. This provides a significant advantage in dealing with systems that have an output entry of relative degree two. In the last decade, NI systems theory has attracted the attention of control theory researchers \cite{Xiong_NI_2010,song2012negative,mabrok2014generalizing,bhikkaji2011negative,bhowmick2017lti}. It has been applied in many fields including nano-positioning control \cite{mabrok2013spectral,das2014mimo,das2014resonant,das2015multivariable} and the control of lightly damped structures \cite{rahman2015design,bhikkaji2011negative}.
		
	Roughly speaking, a square real-rational proper transfer matrix $G(s)$ is said to be NI if it is stable and $j(G(j\omega)-G^*(j\omega))\geq 0$ for all $\omega \geq 0$ \cite{Lanzon_stability_2008,Petersen_feedback_2010,Xiong_NI_2010}. A single-input single-output (SISO) linear NI system has a phase between $0$ to $-180$ degrees for all frequencies $\omega>0$. In other words, a SISO linear NI system has its Nyquist plot below the real axis for all positive frequencies. An NI system can be regarded as a dissipative system, for which the supply rate is the inner product of the system's input and the derivative of the system's output \cite{Xiong_NI_2010,song2012negative}. Under mild assumptions, the positive feedback interconnection of an NI system $G(s)$ and a strictly negative imaginary (SNI) system $G_s(s)$ is asymptotically stable if and only if the DC loop gain of the interconnection is strictly less than unity; i.e., $\lambda_{max}(G(0)G_s(0))<1$ (see \cite{Lanzon_stability_2008}).
	
	NI systems theory was extended to nonlinear systems in \cite{Ghallab_Nonlinear_NI_2018} using the dissipativity property. A system is said to be nonlinear negative imaginary (NNI) if there exists a positive definite storage function $V(x)$ such that $\dot V(x)\leq u^T\dot y$, where $x$, $u$ and $y$ are the state, input and output of the system, respectively (see \cite{shi_robust_identical,ghallab2022negative}). Similar to the linear NI systems theory, asymptotic stability can also be achieved for positive feedback interconnected NNI systems. It is shown in \cite{shi_robust_identical,shi_output_free} that under reasonable assumptions, the interconnection of an NNI system and a nonlinear output strictly negative imaginary (OSNI) system is asymptotic stable. Also,  \cite{ghallab2022negative} shows that the interconnection of an NNI system and a nonlinear weakly strictly negative imaginary (WSNI) system is asymptotically stable, under a different set of assumptions.
	
	Since the phase of a HIGS is in the range of the phase of a typical SISO NI system, it is natural to ask if a HIGS is an NI system. However, given that the dynamics of a HIGS switch between two regions and the state equation in the gain mode is not a differential equation, a more suitable question  is if a HIGS is an NNI system. In this paper, we show that a HIGS is indeed an NNI system, with a common Lyapunov storage function for both the integrator and the gain modes. Motivated by the stability results of positive feedback interconnected NI systems in \cite{Lanzon_stability_2008,Petersen_feedback_2010,shi_robust_identical,shi_output_free,ghallab2022negative}, we investigate the control problem of a SISO linear NI system using a HIGS controller. We show, in this paper, that for any linear time-invariant (LTI) NI system with a minimal realization $(A,B,C)$, there always exists a HIGS such that their positive feedback interconnection is asymptotically stable.
	
	The structure of this paper is as follows: Section \ref{sec:pre} provides essential backgrounds on HIGS and NI systems theory. Section \ref{sec:main} provides the main results of the paper, where we show the NNI property of HIGS and present the stability result of the interconnection of a linear NI system and a HIGS. The stability results are illustrated in Section \ref{sec:example}, where the control problem of a MEMS nanopositioner is considered. Section \ref{sec:conclusion} concludes the paper.
	
	{\bf Notation}: The notation in this paper is standard. $\mathbb R$ denotes the field of real numbers. $\mathbb R^{m\times n}$ denotes the space of real matrices of dimension $m\times n$. $A^T$ and $A^*$ denote the transpose and complex conjugate transpose of a matrix $A$, respectively. $\lambda_{max}(A)$ denotes the largest eigenvalue of a matrix $A$ with real spectrum. $\Re [\cdot]$ is the real part of a complex number. For a symmetric or Hermitian matrix $P$, $P>0\ (P\geq 0)$ denotes the property that the matrix $P$ is positive definite (positive semidefinite) and $P<0\ (P\leq 0)$ denotes the property that the matrix $P$ is negative definite (negative semidefinite).

	%%%%%%%%%%%%%%%%%%%%%%%%%%%%%%%%%%%%%%%%%%%%%%%%%%%%%%%%%%%%%%%%%%%%%%%%%%%	
	\section{PRELIMINARIES}\label{sec:pre}
	In this section, we briefly describe the HIGS and review the main definitions and results in the theory of linear and nonlinear NI systems.
	
	\subsection{Hybrid Integrator-Gain System}
	A SISO hybrid integrator-gain system, $\mathcal{H}$ (HIGS) is represented by the following differential algebraic equations (DAEs)~\cite{Deenen_HIGS_motion_control_2017}:
	\begin{equation}\label{eq:HIGS_DAE}
		\mathcal{H}:
		\begin{cases}
			\dot{x}_h(t) = \omega_h e(t), & \text{if}\, (e(t),u(t),\dot{e}(t)) \in \mathcal{F}_1\\
			x_h(t) = k_he(t), & \text{if}\, (e(t),u(t),\dot{e}(t)) \in \mathcal{F}_2\\
			u(t) = x_h(t)
		\end{cases}
	\end{equation}
where $x_h(t),e(t),u(t) \in \mathbb{R}$ denote the HIGS state, input, and output signals, respectively. For convenience, the variables $x_h(t),e(t),u(t)$ will be denoted as $x_h,e,u$ in what follows. Here, $\dot{e}$ is the time derivative of the input $e$ which is assumed to be continuous and piecewise differentiable. Also, $\omega_h \in [0,\infty)$ and $k_h \in (0, \infty)$ represents the integrator frequency and gain value, respectively. These tunable parameters allow for desired control performance. The sets $\mathcal{F}_1$ and $\mathcal{F}_2 \in \mathbb{R}^3$ determine the HIGS modes of operation; i.e. the integrator and gain modes, respectively. By construction, $\mathcal{F} = \mathcal{F}_1 \cup \mathcal{F}_2$ represents the sector $[0, k_h]$ as~\cite{Deenen_HIGS_motion_control_2017,Achten_HIGS_Skyhook_thesis_2020}
\begin{equation}\label{eq:subspace_F}
	\mathcal{F} = \{ (e,u,\dot{e}) \in \mathbb{R}^3 |\, eu \geq \frac{u^2}{k_h}\},
\end{equation}
and $\mathcal{F}_1$ and $\mathcal{F}_2$ are defined as
	\begin{align}%\label{eq:subspaces_F1_F2}
	\mathcal{F}_1& := \mathcal{F} \setminus \mathcal{F}_2\label{eq:subspaces_F1};\\
	\mathcal{F}_2& := \{(e,u,\dot{e}) \in \mathbb{R}^3 | u = k_he\quad \text{and}\quad  \omega_he^2 > k_he\dot{e}\}\label{eq:subspaces_F2}.
\end{align}

The HIGS is designed to primarily operate in the integrator mode with zero initial condition and output signal; i.e., $x_h(0) = 0$ and $u(0) =0$. A switch to the gain mode is enforced if the corresponding integrator dynamics violate the sector constraint $(e,u,\dot e)\in \mathcal F$. On the boundary of the sector $\mathcal{F}$ where the gain mode is active,  $u(t)$ follows the input behavior and $e(t) = 0$ implies $u(t) = 0$. At switching instances, from the gain mode to the integrator mode, the initial condition is equal to the value in gain mode to ensure a continuous control signal when switching back to the integrator mode.

\subsection{Frequency Analysis of HIGS}
Since the HIGS element contains time-invariant dynamic nonlinearities, it is not possible to obtain the frequency response function of the HIGS using regular frequency-domain Fourier analysis. However, a describing function analysis can be performed  to compute the dominant harmonic of the steady-state response of the system to a single-sinusoid using a first-order Fourier series. Thus, a quasi-linear transfer function mapping from the sinusoidal input $\sin(\omega t)$ to HIGS output $u(t)$ can be determined from~\cite{Heertjes_HIGS_Vibration_2019}
\begin{align}\label{eq:HIGS_describing_function}
	D_h (j\omega) =&\frac{\omega_h}{j\omega}\bigg(\frac{\gamma}{\pi} + j\frac{e^{-j2\gamma} - 1}{2\pi} -4j\frac{e^{-j\gamma}-1}{2\pi}\bigg)\notag \\
	&+k_h\bigg(\frac{\pi - \gamma}{\pi} + j\frac{e^{-j2\gamma}-1}{2\pi}\bigg),
\end{align}
where $\gamma = 2\arctan(\frac{k_h\omega}{\omega_h})$ denotes switching instances. The magnitude and phase of the describing function in (\ref{eq:HIGS_describing_function}) is an approximation of the HIGS filter frequency response~\cite{Achten_HIGS_Skyhook_thesis_2020}. According to (\ref{eq:HIGS_describing_function}), the HIGS acts as a static gain ($k_h$) at low frequencies since $\gamma = 0$. It approaches $\frac{\omega_h}{j\omega}\big(1 + j\frac{4}{\pi}\big)$ at higher frequencies with the following magnitude and phase~\cite{Achten_HIGS_Skyhook_thesis_2020}
\begin{align}%\label{eq:mag_phase_describing_function}
	\lim_{\omega \to \infty} |D(j\omega)| &\approx 1.62 \frac{\omega_h}{\omega};\notag\\
	\lim_{\omega \to \infty} \angle D(j\omega) &\approx -38.1^\circ.\notag
\end{align}
This reveals a phase enhancement of about $52^\circ$ compared to the linear counterpart. Fig.~\ref{fig:HIGS_DF} shows the Bode plot of the HIGS describing function, $D(j\omega)$, where $\omega_c = \omega_h |1 + j4/\pi|$ denotes the cut-off frequency.
\begin{figure}[h!]
	\centering
	\psfrag{Frequency(Hz)}{\hspace{-0.3cm}Frequency(Hz)}
	\psfrag{Phase (deg)}{\small Phase (deg)}
	\psfrag{Magnitude (dB)}{\hspace{-0.1cm}\small Magnitude (dB)}
	\includegraphics[width=\columnwidth]{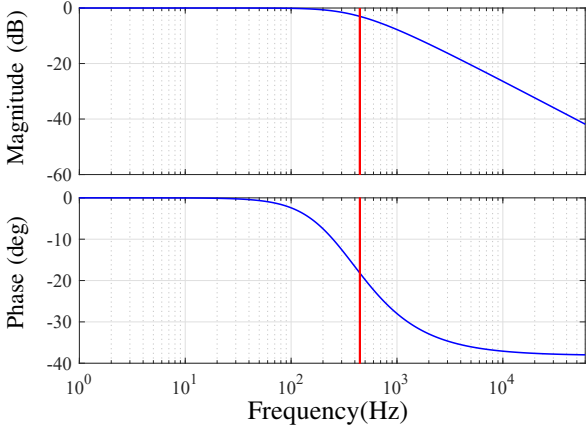}
	\caption{Bode plot of the describing function of a HIGS of the form (\ref{eq:HIGS_DAE}) with $k_h = 1$ and $\omega_h = 600\pi \ \text{rad/s}$.}\label{fig:HIGS_DF}
\end{figure}
%%%%%%%%%%%%%%%%%%%%%%%%%%%%%%%%%%%%%%%%%%%%%%%%%%%%%%%%%%%%%%%%%%%%%%%%%%%
\subsection{Negative Imaginary Systems}
%Negative imaginary (NI) system theory was introduced in~\cite{Lanzon_stability_2008} to address challenges confronting vibration control of flexible structures with collocated force actuators and position sensors. A single-input single output (SISO) linear time-invariant (LTI) system $G(j\omega)$ is said to be negative imaginary if
%\begin{equation}\label{eq:Def_SISO_LTI_NI}
%		j(G(j\omega) - G(j\omega)^T) \geq 0, \quad \omega \in (0,\infty)
%\end{equation}
%In other words, $G(j\omega) $ is NI when its phase lies within $[0^\circ,-180^\circ]$ for all $\omega \in (0,\infty)$ resulting in a Nyquist plot below the real axis for all positive frequencies. This definition can also be extended to strictly negative imaginary (SNI) systems with strict inequality in (\ref{eq:Def_SISO_LTI_NI}).

\begin{definition}{(Negative Imaginary Systems)}\label{def:NI}\cite{Lanzon_stability_2008,Petersen_feedback_2010,Xiong_NI_2010}
A square real-rational proper transfer function matrix $G(s)$ is said to be NI if:

1. $G(s)$ has no poles at the origin and in $\Re [s]>0$;

2. $j[G(j\omega)-G^*(j\omega)]\geq 0$ for all $\omega \in (0,\infty)$ except for values of $\omega$ where $j\omega$ is a pole of $G(s)$;

3. if $j\omega_0$ with $\omega_0\in (0,\infty)$ is a pole of $G(s)$, then it is a simple pole and the residue matrix $K_0=\lim_{s\to j\omega_0}(s-j\omega_0)jG(s)$ is Hermitian and positive semidefinite.
\end{definition}

\begin{lemma}{(NI Lemma)}\cite{Xiong_NI_2010}\label{lem:NI_lemma} Let $(A,B,C,D)$ be a minimal realization of an LTI the transfer function matrix $G(s)$ where $A\in \mathbb{R}^{n\times n}$, $B \in \mathbb{R}^{n\times p}$, $C\in \mathbb{R}^{p\times n}$, and $D\in \mathbb{R}^{p \times p}$. Then, $G(s)$ is NI if and only if:
\begin{enumerate}
    \item $\text{det}(A) \neq 0$, $D = D^T$;
    \item There exists a matrix $Y = Y^T > 0$, $Y \in \mathbb{R}^{n\times n}$, such that
    \begin{equation*}
        AY + YA^T \leq 0,\quad \text{and}\quad B + AYC^T = 0.
    \end{equation*}
\end{enumerate}
\end{lemma}

%The stability analysis of interconnection of NI and SNI systems have been extensively investigated in~\cite{Lanzon_stability_2008,Xiong_NI_2010,Petersen_feedback_2010}. According to NI system theory~\cite{Xiong_NI_2010}, positive feedback interconnection of an NI (SNI) system ($G$) and an SNI (NI) controller ($C$) is internally stable if he DC loop gain is less than unity, i.e. $G(0)C(0) < 1$. This facilitates control design and ensures stability robustness of interconnection NI systems.

Recently, the NI property has been generalized to nonlinear systems and the existing stability results have been extended to a nonlinear setting using Lyapunov and dissipativity theories~\cite{Ghallab_Nonlinear_NI_2018}. Accordingly, a nonlinear system is said to be NNI if the nonlinear system is passive from the input to the derivative of the output. Here, we highlight the main results on NNI systems. These results are used in the ensuing derivations.

Considering a general nonlinear system~\cite{Ghallab_Nonlinear_NI_2018}
\begin{subequations}\label{eq:general_nonlinear_system}
\begin{align}
	\dot{x} &= f(x,u),\\
	y &= h(x),
\end{align}
\end{subequations}
where $f: \mathbb{R}^n \times \mathbb{R}^p \to \mathbb{R}^n$ is a Lipschitz continuous function and $h: \mathbb{R}^n \to \mathbb{R}^p$ is a continuously differentiable function, the following definition describes the NI property of nonlinear systems.
\begin{definition}{~\cite{Ghallab_Nonlinear_NI_2018,shi_robust_identical,ghallab2022negative}}\label{def:nonlinear_NI}
	A system of the form (\ref{eq:general_nonlinear_system}) is said to be an NNI system if there exists a positive definite continuously differentiable storage function $V:\mathbb{R}^n \to \mathbb{R}$ such that
	\begin{equation*}%\label{eq:Lyapunov_dissipative_nonlinear_NI}
		\dot{V}(x(t)) \leq \dot{y}(t)^Tu(t), \quad \forall\, t \geq 0.
	\end{equation*}
\end{definition}
%%%%%%%%%%%%%%%%%%%%%%%%%%%%%%%%%%%%%%%%%%%%%%%%%%%%%%%%%%%%%%%%%%%%%%%%%%%
\section{MAIN RESULTS}\label{sec:main}
In this section, we introduce the main results on the NNI property of the HIGS element and the stability of a positive feedback interconnection of the HIGS and a linear NI system.

\subsection{NNI Property of a HIGS}
We first show a property of the HIGS (\ref{eq:HIGS_DAE}) in Lemma \ref{lemma:F implication}, which is implied by the sector constraint (\ref{eq:subspace_F}). This property will be used later in the stability analysis.
\begin{lemma}\label{lemma:F implication}
Consider a HIGS element with the system model (\ref{eq:HIGS_DAE}). This system satisfies 
\begin{equation*}
ex_h-k_he^2\leq 0,
\end{equation*}
where the equality holds only if $x_h = k_he$.
\end{lemma}
\begin{proof}
Consider the inequality
\begin{equation*}
\left(\sqrt{\frac{1}{k_h}}x_h-\sqrt{k_h}e\right)^2\geq 0,
\end{equation*}
where equality holds only if $x_h=k_he$. Therefore,
\begin{equation*}
\frac{1}{k_h}x_h^2 - 2ex_h + k_he^2 \geq 0.
\end{equation*}
Considering the condition in $\mc F$ as given in (\ref{eq:subspace_F}), this implies that
\begin{equation*}
ex_h-k_he^2\leq \frac{1}{k_h}x_h^2 - ex_h \leq 0,	
\end{equation*}
where equality holds only at $x_h=k_he$.
\end{proof}

Considering the HIGS system in (\ref{eq:HIGS_DAE}), the state and output of HIGS are Lipschitz continuous given a real integrable input and its integrable derivative (see Theorem 4.6.1 in~\cite{Baaij_HIGS_positive_real_systems}). In the following, we show that the HIGS system (\ref{eq:HIGS_DAE}) with input $e(t)$ and output $u(t)$ is NNI.

\begin{theorem}\label{theorem:HIGS_NNI}
    Consider a SISO hybrid integrator-gain system as in (\ref{eq:HIGS_DAE}), then the HIGS is an NNI system from input $e$ to the output $u$ with a positive definite storage function formulated as
	\begin{equation*}%\label{eq:storage_function_HIGS}
		V(x_h) = \frac{1}{2k_h}x_h^2
	\end{equation*}
	satisfying
	\begin{equation}\label{eq:theorem_Vdot_HIGS}
	    \dot{V}(x_h) \leq \dot{u}e.
	\end{equation}
\end{theorem}

\begin{proof}
	The storage function $V(x_h)$ is positive definite since $k_h>0$. Here, we prove that (\ref{eq:theorem_Vdot_HIGS}) holds in both integrator and gain modes. Taking the time derivative of $V$, we have that
	\begin{equation*}%\label{eq:Vdot}
	    \dot{V}(x_h) = \frac{1}{k_h}x_h\dot{x}_h.
	\end{equation*}
	\textit{\textbf{Case 1. }} The HIGS operates in the integrator mode. In this case, according to (\ref{eq:HIGS_DAE}), we have that $\dot{x}_h = \omega_he$. Therefore, $\dot{V}$ is obtained as
    \begin{align}\label{eq:Vdot_integrator_mode}
	    \dot{V}(x_h) &= \frac{1}{k_h}\omega_h eu \nonumber\\
	    &\leq \omega_he^2 = \dot{u}e.
	\end{align}
	where the inequality follows from Lemma \ref{lemma:F implication} and the equation $u=x_h$ in (\ref{eq:HIGS_DAE}).\\
	\textit{\textbf{Case 2. }} The HIGS operates in the gain mode. In this case, the system is in the region $\mathcal F_2$ as given in (\ref{eq:subspaces_F2}), where $u = x_h = k_he$. Therefore,
	\begin{equation}\label{eq:Vdot_gain_mode}
	    \dot{V}(x_h) = \frac{1}{k_h}k_he\dot{x}_h = \dot{u}e.
	\end{equation}
According to (\ref{eq:Vdot_integrator_mode}) and (\ref{eq:Vdot_gain_mode}), the HIGS is an NNI system.
\end{proof}

\subsection{Stability of the Closed-loop Interconnection of a Linear NI System and a HIGS}

Consider the interconnection of a SISO linear NI system $G(s)$ and a HIGS as shown in Fig.~\ref{fig:interconnection}. We analyze the stability of the positive feedback interconnection of a linear NI system and a HIGS in the following.

\begin{figure}[h!]
\centering
\psfrag{in_0}{$r=0$}
\psfrag{in_1}{$u$}
\psfrag{y_1}{$y$}
\psfrag{e}{$e$}
\psfrag{x_h}{$x_h$}
\psfrag{plant}{$G(s)$}
\psfrag{HIGS}{\hspace{-0.3cm} HIGS $\mc H$}
\psfrag{+}{\small$+$}
\includegraphics[width=8cm]{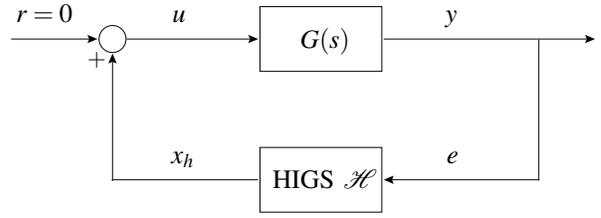}
\caption{Closed-loop interconnection of a linear NI system and a HIGS.}
\label{fig:interconnection}
\end{figure}

Consider a minimal realization of the linear NI system $G(s)$ described as follows:
\begin{subequations}\label{eq:G(s)}
\begin{align}
\dot x =&\ Ax+Bu,\label{eq:G(s) state eq}\\
y =&\ Cx,
\end{align}
\end{subequations}
where $x\in \mathbb R^n$, $u,y \in \mathbb R$ are the state, input and output of the system, respectively. Here, $A\in \mathbb R^{n\times n}$, $B\in \mathbb R^{n\times 1}$ and $C\in \mathbb R^{1\times n}$.

\begin{theorem}\label{theorem:stability of single interconnection}
Consider a SISO linear NI system $G(s)$ with the minimal realization (\ref{eq:G(s)}). There exists a HIGS $\mc H$ of the form (\ref{eq:HIGS_DAE}) such that the closed-loop interconnection of $G(s)$ and $\mc H$ shown in Fig.~\ref{fig:interconnection} is asymptotically stable.
\end{theorem}
\begin{proof}
Since (\ref{eq:G(s)}) is a minimal realization of the linear NI system $G(s)$, then according to Lemma \ref{lem:NI_lemma}, we have that $\det(A)\neq 0$ and there exists $Y=Y^T>0$, $Y\in \mathbb R^{n\times n}$ such that
\begin{equation}\label{eq:NI lemma relations}
	AY+YA^T\leq 0,\qquad \textnormal{and} \qquad B+AYC^T=0.
\end{equation}	
Using Lyapunov's direct method, let the storage function of the closed-loop interconnection be
\begin{align}
W(x,x_h)=&\ \frac{1}{2}x^TY^{-1}x+\frac{1}{2k_h}x_h^2-Cxx_h\notag\\
=&\ \frac{1}{2}\left[\begin{matrix}x^T & x_h\end{matrix}\right]\left[\begin{matrix}Y^{-1}&-C^T\\-C&\frac{1}{k_h}\end{matrix}\right]\left[\begin{matrix}x \\ x_h\end{matrix}\right].\label{eq:sf_for_interconnection}
\end{align}
Using Schur Complement theorem, $W(x,x_h)>0$ for all $(x,x_h)\neq (0,0)$ if
\begin{equation}\label{eq:W pd initial}
	\frac{1}{k_h}-CYC^T>0.
\end{equation}
Using (\ref{eq:NI lemma relations}), we have that $CYC^T=-CA^{-1}B=G(0)$. Then, (\ref{eq:W pd initial}) becomes
\begin{equation*}
	k_hG(0)<1.
\end{equation*}
Using Schur Complement theorem, the positive definiteness of $W(x,x_h)$ also implies that
\begin{equation}\label{eq:W pd condition 2}
	Y^{-1}-k_hC^TC>0,
\end{equation}
which is equivalent to the condition in (\ref{eq:W pd initial}) in this case.
Taking the time derivative of $W(x,x_h)$ defined in (\ref{eq:sf_for_interconnection}), we have
\begin{align}
\dot W(x,x_h) =&\ x^TY^{-1}\dot x+\frac{1}{k_h}x_h\dot x_h-C\dot xx_h-Cx\dot x_h\notag\\
=& \left(x^TY^{-1}-x_hC\right)\dot x+\left(\frac{1}{k_h}x_h-C x\right)\dot x_h\notag\\
=& \left(x^TY^{-1}-uC\right)\dot x+\left(\frac{1}{k_h}x_h-e\right)\dot x_h\notag\\
=&\left(x^TY^{-1}+uB^TA^{-T}Y^{-1}\right)\dot x+\left(\frac{1}{k_h}x_h-e\right)\dot x_h\notag\\
=&\left(x^TA^T+uB^T\right)(A^{-T}Y^{-1})\dot x+\left(\frac{1}{k_h}x_h-e\right)\dot x_h\notag\\
=&\ \frac{1}{2}\dot x^T (A^{-T}Y^{-1}+Y^{-1}A^{-1})\dot x+\left(\frac{1}{k_h}x_h-e\right)\dot x_h,\notag\\
\end{align}
where $u = x_h$ and $e=y=Cx$ are also used. We have that
\begin{align}
	\left(\frac{1}{k_h}x_h-e\right)\dot x_h = \begin{cases}
\left(\frac{1}{k_h}x_h-e\right)\omega_he, & \text{if} (e,x_h,\dot e)\in\mathcal F_1\\
\left(\frac{1}{k_h}x_h-e\right)k_h\dot e, & \text{if} (e,x_h,\dot e)\in \mathcal F_2	
\end{cases}\notag\\
= \begin{cases}
\frac{\omega_h}{k_h}\left(ex_h-k_he^2\right), & \text{if} (e,x_h,\dot e)\in\mathcal F_1\\
\dot e\left(x_h-k_he\right), & \text{if} (e,x_h,\dot e)\in \mathcal F_2.
\end{cases}\label{eq:ineq of HIGS}
\end{align}
In $\mc F_2$, according to (\ref{eq:subspaces_F2}), we have that $x_h=k_he$. Hence, $\left(\frac{1}{k_h}x_h-e\right)\dot x_h=\dot e\left(x_h-k_he\right)=0$. In $\mathcal F_1$, according to Lemma \ref{lemma:F implication}, we have that
\begin{equation*}
	ex_h-k_he^2\leq 0,
\end{equation*}
where equality holds only if $x_h=k_he$. 
Therefore, following from (\ref{eq:ineq of HIGS}), we have that
\begin{equation*}
\left(\frac{1}{k_h}x_h-e\right)\dot x_h \leq 0,
\end{equation*}
and $\left(\frac{1}{k_h}x_h-e\right)\dot x_h = 0$ only if $x_h=k_he$.
We also have that $\frac{1}{2}\dot x^T (A^{-T}Y^{-1}+Y^{-1}A^{-1})\dot x\leq 0$ because $A^{-T}Y^{-1}+Y^{-1}A^{-1}\leq 0$ according to (\ref{eq:NI lemma relations}). Therefore, $\dot W(x,x_h)\leq 0$ and $\dot W(x,x_h)=0$ only if $x_h=k_he$ and $\dot x^T (A^{-T}Y^{-1}+Y^{-1}A^{-1})\dot x=0$. Using LaSalle's invariance principle, $\dot W(x,x_h)$ stays at zero only if $x_h\equiv k_he$ and $\dot x^T (A^{-T}Y^{-1}+Y^{-1}A^{-1})\dot x\equiv 0$. We only consider the case $x\neq 0$ in the following because if $x= 0$ then $x_h = k_he = k_hy = k_hCx = 0$. The condition $x_h\equiv k_he$ implies that
\begin{equation*}%\label{eq:HIGS dot W =0 relation}
	u\equiv k_hy \equiv k_hCx,
\end{equation*}
according to the setting of the interconnection $u=x_h$ and $e=y=Cx$. In this case, the state equation (\ref{eq:G(s) state eq}) of the system $G(s)$ becomes
\begin{equation}
	\dot x = Ax+Bu = Ax+Bk_hCx = (A+k_hBC)x.
\end{equation}
Using (\ref{eq:NI lemma relations}), we have that
\begin{equation}
	A+k_hBC = A-k_hAYC^TC = AY(Y^{-1}-C^TC),
\end{equation}
which is nonsingular according to (\ref{eq:W pd condition 2}) and the non-singularity of the matrices $A$ and $Y$. Therefore, for any $x\neq 0$, we have that $\dot x \neq 0$ and similarly $\ddot x \neq 0$. That is, $\dot x$ can neither remain zero nor a constant vector. In this case, the condition $\dot x^T (A^{-T}Y^{-1}+Y^{-1}A^{-1})\dot x\equiv 0$ implies that $\dot x$ must stay in the null space of the matrix $A^{-T}Y^{-1}+Y^{-1}A^{-1}$. Now we prove that $\dot W(x,x_h)=0$ cannot hold forever. First we prove by contradiction that the HIGS $\mc H$ cannot stay in the integrator mode. Suppose $\mc H$ is in the integrator mode. Then we have that $\dot x_h = \omega_h e$ according to (\ref{eq:HIGS_DAE}). Since we also have that $x_h\equiv k_he$. Then
\begin{equation}\label{eq:ODE in the proof}
	\dot x_h = \omega_h e = k_h \dot e.
\end{equation}
Since $x$ does not remain zero, then according to the observability of the system $G(s)$, we have that $e=y$ does not remain zero. Choosing $\omega_h > 0$, if (\ref{eq:ODE in the proof}) holds for a finite time interval $[t_a,t_b]$ where $t_a<t_b$, then in this time interval,
\begin{equation}
	e(t) = e(t_a)exp(\frac{\omega_h}{k_h}t).
\end{equation}
This means that $y = e$ diverges, which contradicts the fact that the closed-loop system in Fig.~\ref{fig:interconnection} is Lyapunov stable, as has been proved above by showing $\dot W(x,x_h)\leq 0$.
This means that if $\dot W(x,x_h)\equiv 0$, then $\mc H$ cannot stay in the integrator mode. Now we prove that we can always force $\mc H$ to exist the gain mode by choosing suitable HIGS parameters. Suppose $\mc H$ is in the gain mode, then according to (\ref{eq:subspaces_F2}), we have that
\begin{equation}
	\omega_h e^2> k_h e \dot e.
\end{equation}
This inequality cannot be satisfied over time via satisfying $e \dot e<0$ because if so, then $\dot V(x_h) \leq e \dot x_h = k_he\dot e<0$. This implies that the HIGS state $x_h$ will converge to zero and so will $x$ because $x_h = k_h e$, $e = y$ and the system $G(s)$ is observable. Therefore, the HIGS $\mc H$ in the gain mode will eventually satisfy $e\dot e>0$. Since the trajectories of $e$ and $\dot e$ in gain mode are independent of $\omega_h$, then we can choose $\omega_h$ to be sufficiently small in order to violate the condition $\omega_h e^2> k_h e \dot e$ in $\mc F_2$. Hence, $\mc H$ will exist the gain mode and stay in the integrator mode for some finite time. This contradicts the fact that $\mc H$ cannot stay in the integrator mode. Therefore, $\dot W(x,x_h)\equiv 0$ will be violated and $W(x,x_h)$ will decrease monotonically until it reaches zero.

\end{proof}

\begin{comment}
\begin{remark}
The proof of the above theorem shows that the HIGS controller (\ref{eq:HIGS_DAE}) will stabilize the system (\ref{eq:G(s)}) if $k_hG(0)<1$ and $(A+k_hBC)\neq \alpha I$ for all $\alpha \neq 0$.
\end{remark}	
\end{comment}

%%%%%%%%%%%%%%%%%%%%%%%%%%%%%%%%%%%%%%%%%%%%%%%%%%%%%%%%%%%%%%%%%%%%%%%%%%%
\begin{figure}[h!]
	\centering
	\begin{subfigure}[b]{\columnwidth}
	\centering
		\includegraphics[width=0.8\columnwidth]{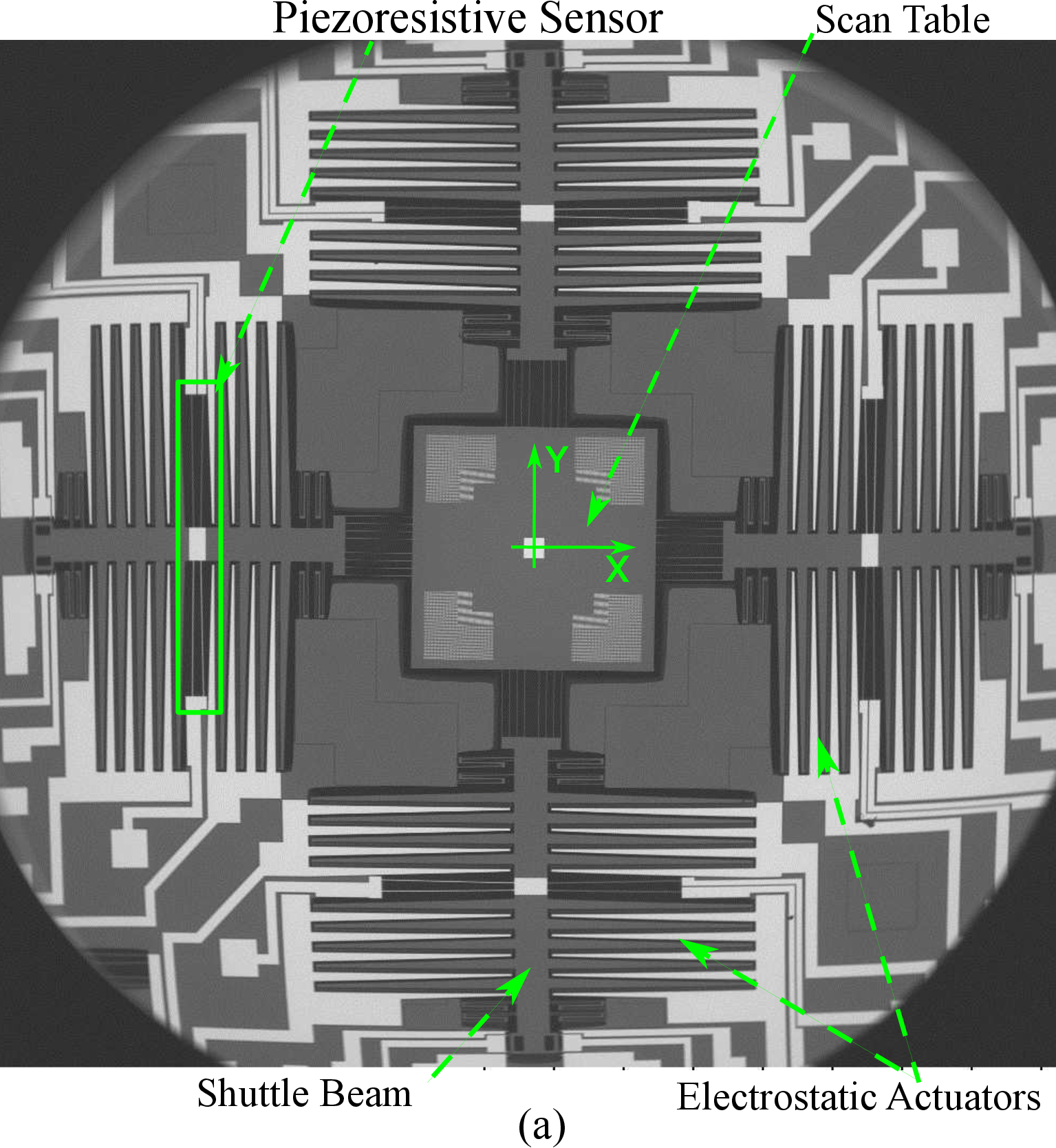}
	\end{subfigure}
	\begin{subfigure}[b]{\columnwidth}
	\centering
		\includegraphics[width=\columnwidth]{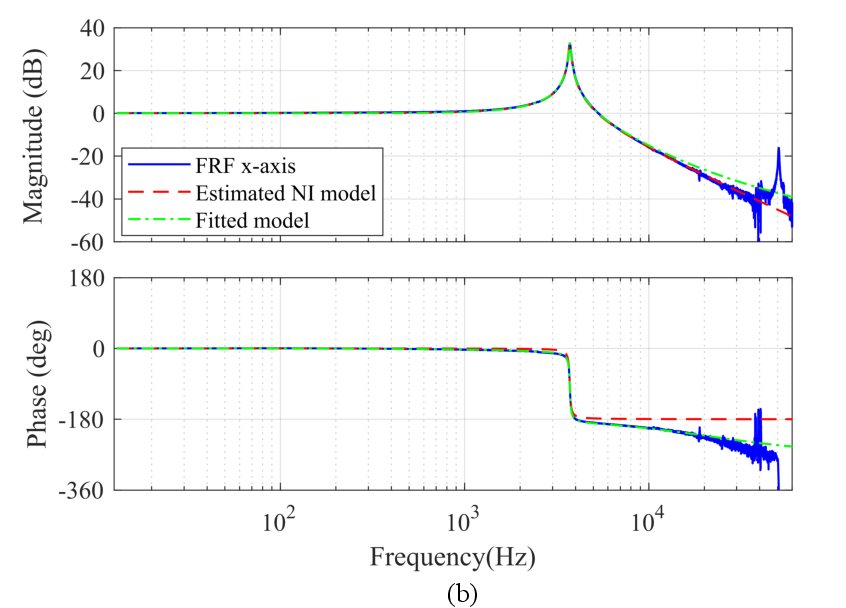}
	\end{subfigure}
	\caption{(a) SEM image of the MEMS nanopositioner reported in~\cite{Maroufi_2DOF_2016}. (b) Frequency response of the MEMS nanopositioner (blue line), along with the approximated NI model (dashed red line), and the fitted model (dot dashed green line). }\label{fig:SEM_image_and_FRF}
\end{figure}

\section{ILLUSTRATIVE EXAMPLE: A MEMS NANOPOSITIONER}\label{sec:example}
In this section, we apply the HIGS controller to a MEMS nanopositioner as a linear NI system and demonstrate the applicability of the stability result.
%Moreover, we show, in simulations, that the positive feedback interconnection of the MEMS nanopositioner and the HIGS element augments damping to the lightly damped mode of the system.

As shown in Fig.~\ref{fig:SEM_image_and_FRF}(a), the 2-DOF MEMS nanopositioner features four electrostatic actuators and on-chip bulk piezoresistive sensors to move the central stage bidirectionally in X- and Y-axis and measure the displacement of the stage, respectively. This device was previously reported in~\cite{Maroufi_2DOF_2016}, where it was employed as a scanning stage for high-speed atomic force microscopy~\cite{Nikooienejad_ILC_2021}. Here, we consider the SISO control design problem for the X-axis of the nanopositioner.

\begin{figure}[h!]
	\centering
	\includegraphics[width=\columnwidth]{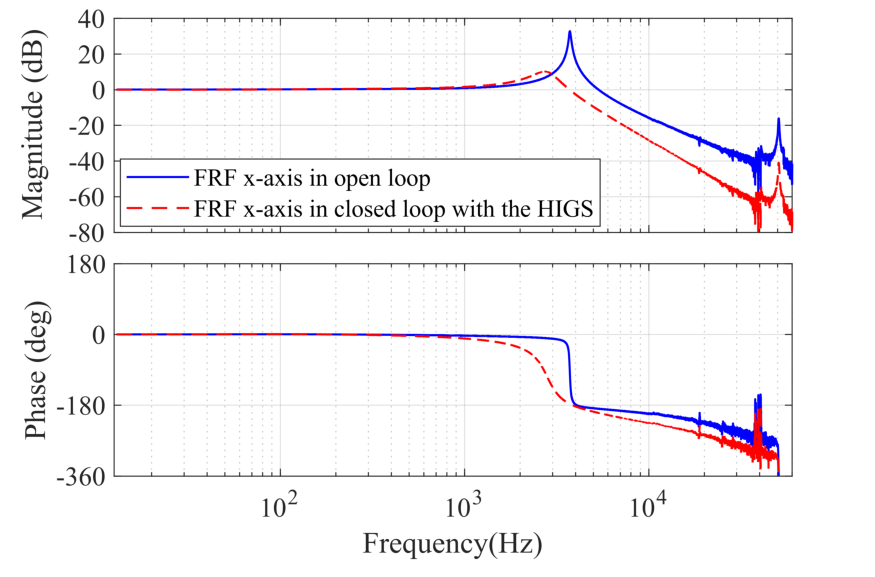}
	\caption{Bode plot of the MEMS nanopositioner in open loop and in a positive feedback interconnection with the HIGS. }\label{fig:CL_FRF}
\end{figure}

Under ideal conditions, a 2-DOF MEMS nanopositioner with collocated actuators and sensors and lightly damped modes can be considered as an NI system~\cite{Xiong_NI_2010}. However, due to the fabrication tolerances and signal conditioning circuit used to read the sensor output voltage, the nanopositioner violates the NI system property beyond a certain frequency where the phase exceeds -180 degrees. Fig.~\ref{fig:SEM_image_and_FRF}(b) shows the frequency response function (FRF) of the nanopositioner from actuation to the sensor output. The fundamental resonance frequency of the nanopositioner along X axis is at $3.725\, \text{kHz}$. We observe that the phase drops beyond $4.064\, \text{kHz}$ and the system violates the NI property as frequency increases. However, the the frequency response rolls off at the rate of $40\, \text{dB/decade}$, and the phase deviation is negligible up to $10\, \text{kHz}$. Therefore, the frequency response of the system can be approximated by a second-order NI model. 
%The unmodeled dynamics of the system introduce a non-minimum phase zero to compensate for the phase drop, resulting in a model that does not meet the conditions in Definition~\ref{def:NI}.
For sake of comparison, both the approximated NI model and the fitted non-NI model are depicted in Fig.~\ref{fig:SEM_image_and_FRF}(b).

Considering the fitted NI model as
\begin{equation*}%\label{eq: MEMS_tf}
    G(s) = \frac{5.493\times 10^8}{s^2 + 541.6 s + 5.465\times 10^8}.
\end{equation*}
The NI property of the system can be verified according to Definition~\ref{def:NI}. The minimal state space realization of the plant is obtained as
\begin{align}%\label{eq:MEMS_ss}
A &= \begin{bmatrix}
-547.571 & -1.6676e4\\
32768 & 0
\end{bmatrix},\quad 
B = \begin{bmatrix}
128\\ 0
\end{bmatrix}\nonumber, \\ C &= \begin{bmatrix} 0 & 130.9727\end{bmatrix},\quad D = 0.\notag
\end{align}
Closed-loop stability of the MEMS nanopositioner in a positive feedback interconnection with the HIGS element can be investigated through Theorem~\ref{theorem:stability of single interconnection}.
\begin{figure}[h!]
	\centering
	\includegraphics[width=\columnwidth]{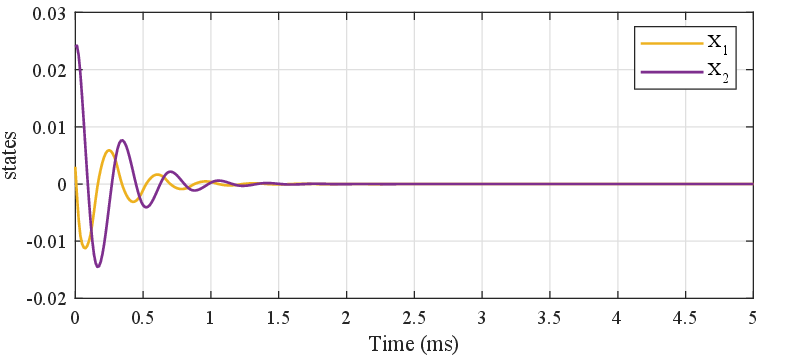}
	\caption{Evolution of state trajectories of the nanopositioner plant under the control of a HIGS controller of the form (\ref{eq:HIGS_DAE}) with $k_h = 0.4939$ and $\omega_h = 1.1705\times 10^4\, \text{rad/s}$}\label{fig:state_trajectories}
\end{figure}

According to Theorem~\ref{theorem:stability of single interconnection}, there exists a HIGS element of the form (\ref{eq:HIGS_DAE}) with $k_h < 0.9868$ that guarantees the closed-loop stability when put in a positive feedback loop with the MEMS nanopositioner. Since $\omega_h$ plays no role in the stability analysis, it can be tuned to achieve the desired performance in time and frequency domains. Accordingly, $k_h = 0.4939$ and $\omega_h = 1.1705\times 10^4\, \text{rad/s}$ are opted for the HIGS. Frequency response of the nanopositioner in closed loop with the HIGS is depicted in Fig.~\ref{fig:CL_FRF}. The closed-loop frequency response is obtained using the describing function of the HIGS and the frequency response data (FRD) model of the nanopositioner. We observe that a substantial damping of about $22\, \text{dB}$ is achieved at the resonance.

We also simulated the positive feedback interconnection of the HIGS and the identified model of the nanopositioner as shown in Fig.~\ref{fig:interconnection}. We captured the state trajectories of the plant in closed loop with zero input and an initial condition of $x_0 = [0.003, 0.024]^T$. The evolution of state trajectories is depicted in Fig.~\ref{fig:state_trajectories}. We observe that the states converge to zero which reveals the stability of the positive feedback interconnection of the MEMS nanopositioner and the selected HIGS element. We also applied a unity step as an input disturbance to the system and analyzed its behavior in closed loop with the HIGS. Fig.~\ref{fig:step_response} demonstrates the step response of the MEMS nanopositioner in open loop and with the HIGS element in a positive feedback interconnection.

\begin{figure}[h!]
	\centering
	\includegraphics[width=\columnwidth]{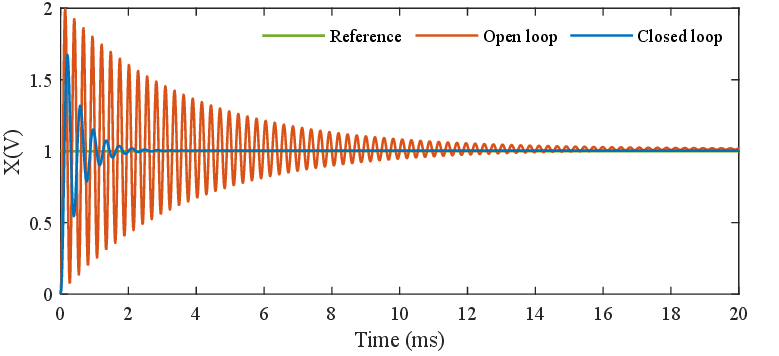}
	\caption{Step response of the MEMS nanopositioner in open loop and with the HIGS. }\label{fig:step_response}
\end{figure}
%%%%%%%%%%%%%%%%%%%%%%%%%%%%%%%%%%%%%%%%%%%%%%%%%%%%%%%%%%%%%%%%%%%%%%%%%%%
\section{CONCLUSION}\label{sec:conclusion}
	This paper investigates an application of the hybrid integrator-gain systems to the control of SISO linear NI systems. The NNI property of the HIGS is analyzed using the time-domain definition of NNI systems. The positive feedback interconnection of a HIGS and a linear NI system is also proved to be asymptotically stable. To illustrate the stability results, a HIGS element is designed and applied to the identified model of a MEMS nanopositioner with collocated force actuators and position sensor. Simulation results confirm the closed-loop stability and demonstrate the efficacy of the HIGS element as an NNI controller. Future work will focus on the NI property of the cascade of HIGS elements and the stability for the interconnection of a linear NI system and the cascaded hybrid system. The control design problem with the HIGS in multiple-input multiple-output framework will be further examined.
%%%%%%%%%%%%%%%%%%%%%%%%%%%%%%%%%%%%%%%%%%%%%%%%%%%%%%%%%%%%%%%%%%%%%%%%%
	\balance

	\bibliographystyle{ieeetr}

\begin{thebibliography}{10}

\bibitem{Deenen_HIGS_motion_control_2017}
D.~Deenen, M.~Heertjes, W.~Heemels, and H.~Nijmeijer, ``Hybrid integrator
  design for enhanced tracking in motion control,'' in {\em 2017 American
  Control Conference (ACC)}, pp.~2863--2868, 2017.

\bibitem{Heertjes_HIGS_Vibration_2019}
M.~Heertjes, S.~van~den Eijnden, B.~Sharif, M.~Heemels, and H.~Nijmeijer,
  ``Hybrid integrator-gain system for active vibration isolation with improved
  transient response,'' {\em IFAC-PapersOnLine}, vol.~52, no.~15, pp.~454--459,
  2019.
\newblock 8th IFAC Symposium on Mechatronic Systems MECHATRONICS 2019.

\bibitem{Clegg_1958}
J.~C. Clegg, ``A nonlinear integrator for servomechanisms,'' {\em Transactions
  of the American Institute of Electrical Engineers, Part II: Applications and
  Industry}, vol.~77, no.~1, pp.~41--42, 1958.

\bibitem{Zaccarian_FORE_2005}
L.~Zaccarian, D.~Nesic, and A.~Teel, ``First order reset elements and the clegg
  integrator revisited,'' in {\em Proceedings of the 2005, American Control
  Conference, 2005.}, pp.~563--568 vol. 1, 2005.

\bibitem{Eijnden_HIGS_motion_control_2018}
S.~Van~den Eijnden, Y.~Knops, and M.~Heertjes, ``A hybrid integrator-gain based
  low-pass filter for nonlinear motion control,'' in {\em 2018 IEEE Conference
  on Control Technology and Applications (CCTA)}, pp.~1108--1113, 2018.

\bibitem{Gruntjens_HIGS_Lens_motion_2019}
K.~Gruntjens, M.~Heertjes, S.~Van~Loon, N.~Van De~Wouw, and W.~Heemels,
  ``Hybrid integral reset control with application to a lens motion system,''
  in {\em 2019 American Control Conference (ACC)}, pp.~2408--2413, 2019.

\bibitem{Achten_HIGS_Skyhook_thesis_2020}
A.~S.~P., ``{HIGS}-based skyhook damping design of a multivariable vibration
  isolation system,'' Master's thesis, Eindhoven University of Technology,
  2020.

\bibitem{Baaij_HIGS_positive_real_systems}
R.~de~{Baaij}, ``Non-linear dynamic control for positive-real systems,''
  February 2021.

\bibitem{Eijnden_HIGS_Overshoot_limitation_2020}
S.~J. A.~M. van~den Eijnden, M.~F. Heertjes, W.~P. M.~H. Heemels, and
  H.~Nijmeijer, ``Hybrid integrator-gain systems: A remedy for overshoot
  limitations in linear control?,'' {\em IEEE Control Systems Letters}, vol.~4,
  no.~4, pp.~1042--1047, 2020.

\bibitem{Vanloon_stability_reset_system_2017}
S.~{van Loon}, K.~Gruntjens, M.~Heertjes, N.~{van de Wouw}, and W.~Heemels,
  ``Frequency-domain tools for stability analysis of reset control systems,''
  {\em Automatica}, vol.~82, pp.~101--108, 2017.

\bibitem{Loon_stability_reset_2017}
S.~{van Loon}, K.~Gruntjens, M.~Heertjes, N.~{van de Wouw}, and W.~Heemels,
  ``Frequency-domain tools for stability analysis of reset control systems,''
  {\em Automatica}, vol.~82, pp.~101--108, 2017.

\bibitem{Eijnden_frequency_stability_HIGS_2021}
S.~J. A.~M. Van Den~Eijnden, M.~F. Heertjes, W.~P. M.~H. Maurice~Heemels, and
  H.~Nijmeijer, ``Frequency-domain tools for stability analysis of hybrid
  integrator-gain systems,'' in {\em 2021 European Control Conference (ECC)},
  pp.~1895--1900, 2021.

\bibitem{Deenen_projection_based_2022}
D.~A. Deenen, B.~Sharif, S.~{van den Eijnden}, H.~Nijmeijer, M.~Heemels, and
  M.~Heertjes, ``Projection-based integrators for improved motion control:
  Formalization, well-posedness and stability of hybrid integrator-gain
  systems,'' {\em Automatica}, vol.~133, p.~109830, 2021.

\bibitem{Lanzon_stability_2008}
A.~{Lanzon} and I.~R. {Petersen}, ``Stability robustness of a feedback
  interconnection of systems with negative imaginary frequency response,'' {\em
  IEEE Transactions on Automatic Control}, vol.~53, no.~4, pp.~1042--1046,
  2008.

\bibitem{Petersen_feedback_2010}
I.~R. {Petersen} and A.~{Lanzon}, ``Feedback control of negative-imaginary
  systems,'' {\em IEEE Control Systems Magazine}, vol.~30, no.~5, pp.~54--72,
  2010.

\bibitem{preumont2018vibration}
A.~Preumont, {\em Vibration control of active structures: an introduction},
  vol.~246.
\newblock Springer, 2018.

\bibitem{halim2001spatial}
D.~Halim and S.~O.~R. Moheimani, ``Spatial resonant control of flexible
  structures-application to a piezoelectric laminate beam,'' {\em IEEE
  Transactions on Control Systems Technology}, vol.~9, no.~1, pp.~37--53, 2001.

\bibitem{pota2002resonant}
H.~Pota, S.~O.~R. Moheimani, and M.~Smith, ``Resonant controllers for smart
  structures,'' {\em Smart Materials and Structures}, vol.~11, no.~1, p.~1,
  2002.

\bibitem{brogliato2007dissipative}
B.~Brogliato, R.~Lozano, B.~Maschke, and O.~Egeland, ``Dissipative systems
  analysis and control,'' {\em Theory and Applications}, vol.~2, 2007.

\bibitem{shi2021necessary}
K.~Shi, I.~R. Petersen, and I.~G. Vladimirov, ``Necessary and sufficient
  conditions for state feedback equivalence to negative imaginary systems,''
  {\em arXiv preprint arXiv:2109.11273}, 2021.

\bibitem{Xiong_NI_2010}
J.~{Xiong}, I.~R. {Petersen}, and A.~{Lanzon}, ``A negative imaginary lemma and
  the stability of interconnections of linear negative imaginary systems,''
  {\em IEEE Transactions on Automatic Control}, vol.~55, no.~10,
  pp.~2342--2347, 2010.

\bibitem{song2012negative}
Z.~Song, A.~Lanzon, S.~Patra, and I.~R. Petersen, ``A negative-imaginary lemma
  without minimality assumptions and robust state-feedback synthesis for
  uncertain negative-imaginary systems,'' {\em Systems \& Control Letters},
  vol.~61, no.~12, pp.~1269--1276, 2012.

\bibitem{mabrok2014generalizing}
M.~A. Mabrok, A.~G. Kallapur, I.~R. Petersen, and A.~Lanzon, ``Generalizing
  negative imaginary systems theory to include free body dynamics: Control of
  highly resonant structures with free body motion,'' {\em IEEE Transactions on
  Automatic Control}, vol.~59, no.~10, pp.~2692--2707, 2014.

\bibitem{bhikkaji2011negative}
B.~Bhikkaji, S.~O.~R. Moheimani, and I.~R. Petersen, ``A negative imaginary
  approach to modeling and control of a collocated structure,'' {\em IEEE/ASME
  Transactions on Mechatronics}, vol.~17, no.~4, pp.~717--727, 2011.

\bibitem{bhowmick2017lti}
P.~Bhowmick and S.~Patra, ``On {LTI} output strictly negative-imaginary
  systems,'' {\em Systems \& Control Letters}, vol.~100, pp.~32--42, 2017.

\bibitem{mabrok2013spectral}
M.~A. Mabrok, A.~G. Kallapur, I.~R. Petersen, and A.~Lanzon, ``Spectral
  conditions for negative imaginary systems with applications to
  nanopositioning,'' {\em IEEE/ASME Transactions on Mechatronics}, vol.~19,
  no.~3, pp.~895--903, 2013.

\bibitem{das2014mimo}
S.~K. Das, H.~R. Pota, and I.~R. Petersen, ``A {MIMO} double resonant
  controller design for nanopositioners,'' {\em IEEE Transactions on
  Nanotechnology}, vol.~14, no.~2, pp.~224--237, 2014.

\bibitem{das2014resonant}
S.~K. Das, H.~R. Pota, and I.~R. Petersen, ``Resonant controller design for a
  piezoelectric tube scanner: A mixed negative-imaginary and small-gain
  approach,'' {\em IEEE Transactions on Control Systems Technology}, vol.~22,
  no.~5, pp.~1899--1906, 2014.

\bibitem{das2015multivariable}
S.~K. Das, H.~R. Pota, and I.~R. Petersen, ``Multivariable negative-imaginary
  controller design for damping and cross coupling reduction of
  nanopositioners: a reference model matching approach,'' {\em IEEE/ASME
  Transactions on Mechatronics}, vol.~20, no.~6, pp.~3123--3134, 2015.

\bibitem{rahman2015design}
M.~A. Rahman, A.~Al~Mamun, K.~Yao, and S.~K. Das, ``Design and implementation
  of feedback resonance compensator in hard disk drive servo system: A mixed
  passivity, negative-imaginary and small-gain approach in discrete time,''
  {\em Journal of Control, Automation and Electrical Systems}, vol.~26, no.~4,
  pp.~390--402, 2015.

\bibitem{Ghallab_Nonlinear_NI_2018}
A.~G. Ghallab, M.~A. Mabrok, and I.~R. Petersen, ``Extending negative imaginary
  systems theory to nonlinear systems,'' in {\em 2018 IEEE Conference on
  Decision and Control (CDC)}, pp.~2348--2353, 2018.

\bibitem{shi_robust_identical}
K.~Shi, I.~G. Vladimirov, and I.~R. Petersen, ``Robust output feedback
  consensus for networked identical nonlinear negative-imaginary systems,''
  {\em IFAC-PapersOnLine}, vol.~54, no.~9, pp.~239--244, 2021.

\bibitem{ghallab2022negative}
A.~G. Ghallab and I.~R. Petersen, ``Negative imaginary systems theory for
  nonlinear systems: A dissipativity approach,'' {\em arXiv preprint
  arXiv:2201.00144}, 2022.

\bibitem{shi_output_free}
K.~Shi, I.~R. Petersen, and I.~G. Vladimirov, ``Output feedback consensus for
  networked heterogeneous nonlinear negative-imaginary systems with free body
  motion,'' {\em arXiv preprint arXiv:2011.14610v2}, 2020.

\bibitem{Maroufi_2DOF_2016}
M.~{Maroufi} and S.~O.~R. {Moheimani}, ``A {2DOF} {SOI-MEMS} nanopositioner
  with tilted flexure bulk piezoresistive displacement sensors,'' {\em IEEE
  Sensors Journal}, vol.~16, no.~7, pp.~1908--1917, 2016.

\bibitem{Nikooienejad_ILC_2021}
N.~Nikooienejad, M.~Maroufi, and S.~O.~R. Moheimani, ``Iterative learning
  control for video-rate atomic force microscopy,'' {\em IEEE/ASME Transactions
  on Mechatronics}, vol.~26, no.~4, pp.~2127--2138, 2021.

\end{thebibliography}
	
\end{document}